\documentclass[nonacm,sigconf]{acmart}

\usepackage{macros}
\begin{document}

\title{Exact computations with quasiseparable matrices}


\author{Cl\'ement Pernet}
\orcid{0001-6970-0417}
\affiliation{
  \institution{Grenoble INP, Univ. Grenoble Alpes}
  \department{CNRS, LJK, UMR 5224}
  \city{Grenoble}
  \country{France}
}

\author{Hippolyte Signargout}
\affiliation{
\institution{ENS de Lyon, U. Lyon, CNRS, Inria, UCBL, LIP UMR 5668 Lyon and LJK UMR 5224 Grenoble, France}
  \country{}
}

\author{Gilles Villard}
\affiliation{
  \institution{CNRS, U. Lyon, Inria, ENS de Lyon, UCBL, LIP UMR 5668}
  \city{Lyon}
  \country{France}
}


\begin{abstract}
  Quasi-separable matrices are a class of rank-structured matrices widely used in numerical linear
  algebra and of growing interest  in computer algebra, with applications in e.g. the linearization of
  polynomial matrices. Various representation formats exist for these matrices that have rarely been
  compared. 

  We show how the most central formats \SSS and \HSS can be adapted to symbolic computation, where
  the exact rank replaces threshold based numerical ranks. 
  We clarify their links and compare them with the \Bruhat format. To this end, we
  state their space and time cost estimates based on fast matrix multiplication, and compare them,
  with their leading constants. The comparison is supported by software experiments.
  
  We make further progresses for the \Bruhat format, for which we give a generation algorithm,
  following a Crout elimination scheme, which specializes into fast algorithms for the construction
  from a sparse matrix or from the sum of \Bruhat representations.


  \end{abstract}


\keywords{Quasiseparable matrix, SSS, HSS, Bruhat generator}

\maketitle

\section{Introduction}
Quasiseparable matrices arise frequently in various problems of numerical analysis and are becoming
increasingly important in computer algebra, e.g. by their application to handle linearizations of polynomial
matrices~\cite{BEG17}.
Structured representations for these matrices and their generalisations have been widely studied but
to our knowledge they have not been 
compared in detail with each other. 
In this paper we aim to adapt \texttt{SSS}~\cite{Eidelman1999OnAN}
and \texttt{HSS}   \cite{changu03, lyons2005fast}, two of the
most prominent formats of numerical analysis to exact computations
and compare them theoretically and experimentally to the \Bruhat 
format \cite{PS18}.
These formats all have linear storage size in both the dimension and the structure parameter.
We do not investigate the Givens weight representation~\cite{DeBa08} as it strongly relies on
orthogonal transformations in \CC, which transcription in the algebraic setting is more challenging.
See \cite{Vandebril2005ABO, VVBM08, Hackbusch2015HierarchicalMA} for
an extensive bibliography on
computing with
quasiseparable matrices. 
\begin{definition}
  An \(n \times n\) matrix \(A\) is \(s\)-quasiseparable if for all \(k \in \intset{1,n}\), 
  \(\rank(\submat{A}{1..k}{k+1..n})\leq  s\) and 
  \(\rank(\submat{A}{k+1..n}{1..k})\leq  s\).
\end{definition}

\smallskip 
\noindent 
{\em Complexity bound notation.}
We consider matrices over an abstract commutative field $\K$, and count arithmetic operations in $\K$.
Our detailed comparison of formats  aims in particular to determine the asymptotic multiplicative
constants, an insightful measure on the algorithm's behaviour in pratice. 
In this regard, we will use the leading term in the complexities as the measure for our comparison: 
namely a function \(\Time{XXX}(n,s)\) such that the number of field operations for
running Algorithm XXX
with parameters \(n,s\) is \(\Time{XXX}(n,s) + o(\Time{XXX}(n,s))\)
asymptotically in \(n\) and \(s\).
We proceed similarly for the space cost bounds with the notation \(\Space{XXX}(n,s)\).
We denote by \(\omega\) a
feasible exponent for square matrix multiplication, and \(C_\omega\) the
corresponding leading constant;  namely, using above notation, \(\Time{MM}(n)=C_\omega n^\omega\),
where \texttt{MM} corresponds to the operation \(C = C + AB\) with \(A,B,C\in\K^{n \times n}\).
The straightforward generalization  gives
\(\Time{MM}(m,k,n)=C_\omega mnk \min(m,k,n)^{\omega -3}\) for the product of 
an \(m\times k\) by a \(k\times n\) matrix.

\subsection{Rank revealing factorizations}

Space efficient representations for quasiseparable matrices rely on
 rank revealing factorizations: a rank \(r\)  matrix \(A\in\K^{m\times n}\)  is represented by two 
 matrices \(L\in\K^{m\times r}
 R\in\K^{r\times n}\) such that \(A=LR\).
%
%
In exact linear algebra,  such  factorizations are usually computed using Gaussian elimination, such
as PLUQ, CUP, PLE, CRE decompositions~\cite{JPS13,DPS17,Sto00}, which we will generically denote by \PLUQ. 

Cost estimates of the above factorization algorithms are either given as \(\bigO{mnr^{\omega-2}}\) or with explicit leading constants 
 \(\Time{RF}(m,n,r)=K_\omega n^\omega\) under genericity assumptions: \(m=n=r\)
and generic rank profile~\cite{JPS13,DPS17}.
We refer to~\cite{PSV23} for an  analysis in the non-generic case of the leading constants in the cost of the two main
variants of divide and conquer Gaussian elimination algorithms.
We may therefore assume that \(\Time{RF}(m,n,r)=C_\texttt{RF}mnr^{\omega-2}\)
for a  constant $C_\texttt{RF}$, for \(\omega \geq 1+\log_2 3\), which is the case for all pratical
matrix multiplication algorithm. Note that for \(\omega=3\), these costs are both equal to
\(2mnr\). 
Unfortunately, the non-predictable rank distribution among the blocks being processed leads to 
an over-estimation of some intermediate costs which forbids tighter constants
(i.e. interpolating the known one \(K_3=2/3\) in the generic case).
The algorithms presented here still carry on for smaller values of \(\omega\), but we chose to skip
the more complex derivation of estimates on their leading constants for the sake of clarity.


Our algorithms for \SSS and \HSS can use any rank revealing factorization. 
On the other hand, the \Bruhat format requires one revealing the additional information of the rank profile matrix,
 e.g. the CRE decompositions used here (See~\cite{DPS17}).

\begin{theorem}[\cite{MaHe07,DPS17}]\label{th:cre}
  Any rank \(r\) matrix \(A\in\K^{m\times n}\) has a CRE decomposition \(A=CRE\) where \(C\in
  \K^{m\times r}\) and \(E\in\K^{r\times n}\) are in column and row echelon form,
  and \(R\in\K^{r\times r}\) is a permutation matrix.
\end{theorem}
The costs we give in relation to \Bruhat generator therefore rely on constants \(C_\texttt{RF}\)
from factorizations allowing to produce a CRE decomposition, like the ones in~\cite{PSV23}.
\begin{table*}[ht]
    \caption{Summary of operation and storage costs}\label{tab}
  \begin{tabular}{l|ccc|ccc}
    \toprule
&    \multicolumn{3}{c|}{\(\omega\)} & \multicolumn{3}{c}{\(\omega=3\)}\\
     & \SSS & \HSS & \Bruhat   & \SSS & \HSS & \Bruhat \\
    \midrule
    Storage               & \(7ns\) & \(18ns\) & \(4ns\)& \(7ns\) & \(18ns\) & \(4ns\)\\
    Gen. from Dense & \(2C_\texttt{RF}n^2s^{\omega-2}\) &\(2^{\omega}C_\texttt{RF}n^2s^{\omega-2}\) & \(C_\texttt{RF}n^2s^{\omega-2}\) & \(4n^2s\) &\(16n^2s\) & \(2n^2s\) \\ 
    \(\times\) Dense block vector\((n\times v)\)
    &  \(7C_\omega nsv^{\omega-2}\) & \(18C_\omega nsv^{\omega-2}\) & \(8C_\omega
    nsv^{\omega-2}\) &  \(14 nsv\) & \(36 nsv\) & \(16 nsv\) \\
    Addition & \((10+2^\omega)C_\omega ns^{\omega-1}\) &&\(\paren{\frac{9 \cdot 2^{\omega - 2} - 8}{2^{\omega - 2} - 1}\cw
      + 2\crf} n s^{\omega - 1} \log n/s\)& \(36ns^2\) & & \(24ns^2 \log n/s\)\\
    Product & \((31+2^\omega)C_\omega ns^{\omega-1}\) &&& \(78 n s^2\) \\
    \bottomrule
    \end{tabular}
  \end{table*}

\subsection{Contributions}
In \cref{sec:gen} we define the \SSS, \HSS and \Bruhat formats.
We then adapt algorithms operating with \HSS 
and \SSS generators 
from the literature to the exact context.
The \HSS generation algorithm is given in a new iterative version
and the \SSS product algorithm has an improved cost.
We focus for \SSS on basic bricks on which other operations can be built.
This opens the door to adaptation of
fast algorithms for inversion and system solving \cite{CDGPV03, EG05bsss, CDGPSVW05}
and format modeling operations such as
merging, splitting and model reduction \cite{CDGPV03}.
In \cref{sec:dtb} we give a generic \Bruhat generation algorithm from which we derive new fast algorithms for
the generation from a sparse matrix and from a sum of matrices in \Bruhat form.

\Cref{tab} displays the best cost estimates for differents operations on an \(n \times n\) \(s\)-quasi-separable matrix in the three formats presented in the paper.
The best and optimal storage  size is reached by the \Bruhat format which also has the fastest generator computation algorithm.
However, this is not reflected in the following operation costs as applying a quasiseparable matrix to a dense matrix is least expensive with an \SSS generator
and addition and product of \(n \times n\) matrices given in \Bruhat form is super-linear in \(n\).
We notice in \cref{prop:hcost} that \HSS is twice as expensive as \SSS and gives no advantage in our context.
We thus stop the comparison at the generator computation.
We still give in \cref{tab} the cost of quasiseparable \(\times\) dense product which is proportional to the generator size \cite{lyons2005fast}.
We complete this analysis with experiments showing that despite slightly worse
asymptotic cost estimates,  \SSS performs better than \Bruhat in practice for the construction in
\Cref{sec:ecd} and the product by a dense block vector in \Cref{sec:eca}.


\section{Presentation of the formats}
\label{sec:gen}
\subsection{\SSS generators} \label{sec:SSSgen}
Introduced in~\cite{Eidelman1999OnAN}, \SSS generators were later improved independently in
\cite{EG05bsss} and \cite{CDGPV03}  using block-versions, which we  present here.
In particular, the space was improved from \(\bigO{ns^2}\) to \(\bigO{ns}\).


An \(s\)-quasiseparable matrix is sliced following a grid of \(s \times s\) blocks.
Blocks on, over and under the diagonal are treated separately.
On one side of the diagonal, each block is defined by a product depending on its row
(left-most block of the product), its column (right-most block), and its distance to
the diagonal
(number of blocks in the product).

\begin{definition} \label{def:SSS}
  Let \(A = \left[\begin{smallmatrix}A_{1,1} & \cdots & A_{1,N} \\
    \vdots && \vdots \\
    A_{N,1} & \cdots & A_{N,N} \end{smallmatrix}\right] \in \K^{n \times n}\)
  with \(t \times t\) blocks \(A_{i,j}\) for \(i,j <N\) and \(N = \ceil{n/t}\). 
\(A\) is given in sequentially semi-separable format of
order \(t\) (\(t\)-\SSS) if it is given by the \(t \times t\) matrices
\(\paren{P_i, V_i}_{i\in\intset{2,N}}\), \(
  \paren{Q_i, U_i}_{i\in\intset{1, N - 1}},
  \paren{R_i, W_i}_{i\in\intset{2,N - 1}},
  \paren{D_i}_{i\in\intset{1,N} }\) s.t.
  \begin{equation}
    \label{eq:def}
  A_{i,j} = \left\{\begin{matrix} P_i  R_{i - 1} \dots R_{j + 1} Q_j & \text{if } i > j \\
  D_i & \text{if } i = j \\
  U_i W_{i + 1} \dots W_{j - 1} V_j & \text{otherwise}
  \end{matrix}\right.
  \end{equation}
\end{definition}


\begin{proposition}
  Any \(n \times n\) \(s\)-quasiseparable matrix has an \(s\)-\SSS representation. It uses   \(\Space{\SSS}(n,s)=7ns\) field elements.
\end{proposition}
 \begin{proof}
   Direct consequence of \cref{prop:dts}.
   \end{proof}

\subsection{\HSS generators} \label{sec:HSSSgen}
The \HSS format was first introduced in \cite{CGP06ULV}, although the idea originated with the {\it uniform \(\mathcal H\)-matrices}
of \cite{Hackbusch1999ASM} and in more details with the {\it \(\mathcal H^2\)-matrices} of \cite{HKS99H2}, with algorithms relying on \cite{Starr91}.
The \(\mathcal H^2\) format is slightly different from \HSS, more details in \cite{Hackbusch2015HierarchicalMA}.

The format is close to \SSS (see \cref{prop:altdef}) as the way of defining blocks is similar. Yet, 
the slicing grid is built recursively and the definition of blocks product depends on the path to follow in the recursion tree.
Also, both sides of the diagonal are treated jointly and the format is therefore less compact, which
as will be shown 
makes \HSS less efficient.

The structure is complex and notations differ in the literature.
We made the following choices:
we avoid the recursive tree definition inherited from the Fast Multipole Method \cite{CGP06ULV} and
 thus only consider constant-depth recursive block divisions.
  We made this choice to focus on linear algebra and quasiseparable matrices with no pre-requisites
  (no notion of where the rank is).
 For the same reason we focus on uniform subdivisions.
  Most literature on \HSS  
  uses non-uniform grids in order to adapt to matrices with a structure
within the quasiseparable rank structure \cite{CGP06ULV}.
Despite being more general, this adds confusion which is not needed in our case.

We use a notation similar to \cite{XCGL10HSS} with transition matrices.
  
\begin{definition}
  Let \(A \in \K^{n \times n}\) and the uniform block divisions
  \begin{equation}
    \label{eq:kdiv}
A = \left[\begin{smallmatrix}A_{k;1,1} & \cdots & A_{k;1,2^k} \\
    \vdots && \vdots \\
    A_{k;2^k,1} & \cdots & A_{k;2^k, 2^k} \end{smallmatrix}\right].
\end{equation}
\(A\) is given in hierarchically semi-separable format of
order \(t\) (\(t\)-\HSS) if it is given by the \(t \times t\) matrices 
\(
  \paren{U_{K;i}, V_{K;i}, D_i}_{i\in\intset{1,N}}\), \(
  \paren{R_{k;i}, W_{k;i}}_{\substack{k\in\intset{2,K}\\ i\in\intset{1,2^k}}}\) and
  \(\paren{B_{k;i}}_{\substack{k\in\intset{1,K}\\ i\in\intset{1,2^k}}}\)
  with \(N = \ceil{n/t}\) and \(K \geq\log N\)
  such that for \(i \in \intset{1,N}\), 
 \(   A_{K;i,i} = D_i\)
  and if we define recursively for \(k\) from  \(K - 1\) to \(1\) and \(i \in \intset{1,{2^k}}\),
\(    U_{k;i} = \begin{bmatrix} U_{k + 1; 2i - 1}R_{k + 1; 2i - 1} \\
U_{k + 1; 2i}R_{k + 1; 2i} \end{bmatrix}\) and \(V_{k;i} = \begin{bmatrix} W_{k + 1; 2i - 1}V_{k + 1; 2i - 1} &
      W_{k + 1; 2i}V_{k + 1; 2i} \end{bmatrix}\) then
\begin{equation}\label{eq:HSS}
  \begin{array}{rcl}
    A_{k;2i - 1, 2i} &=& U_{k;2i - 1}B_{k;2i - 1}V_{k; 2i} \\
    A_{k;2i, 2i - 1} &=& U_{k;2i}B_{k;2i}V_{k; 2i - 1}
    \end{array}
    \end{equation}
\end{definition}
The \HSS generator can be seen as a recursive \SSS generator with two differences :
the use of the \(B\) matrices, and the distribution of the translation matrices.
The similarity is made clear in \cref{prop:altdef}.
\begin{proposition}
  \label{prop:altdef}
Let \(U_{K;i}, V_{K;i}, D_i, R_{k;i}, W_{k;i}, B_{k;i}\) for appropriate \(k \leq K, i \leq 2^k\)
  a \(t\)-\HSS generator for \(A\).
  Let \(I,J \in \intset{1,2^K}\) and \(k\) the highest level of recursion for which \(A_{K;I,J}\) is not included in a diagonal block.
For \(i_1 = \lfloor I / 2^{K - k - 1}\rfloor, i_0 = \lfloor I / 2^{K - k}\rfloor\) and \(j_1 =
    \lfloor J / 2^{K - k - 1}\rfloor\) we have
    \begin{equation}
    \label{eq:altdef}
    A_{K ; I, J} =
    U_{K ; I} R_{K ; I} ... R_{k + 1;  i_1} B_{k ;  i_0}
    W_{k + 1;  j_1} ... W_{K ; J} V_{K ; J}.
    \end{equation}
\end{proposition}

\begin{proof}
  By induction on~\Cref{eq:HSS}.
\end{proof}

\begin{proposition}
  \label{prop:hcost}
   Any \(n \times n\) \(s\)-quasiseparable matrix has a \(2s\)-\HSS representation.
   This is the optimal block parameter and the representation uses \(\Space{HSS}(n,s)=18ns\) field elements.
 \end{proposition}
 \begin{proof}
   Consequence of \cref{prop:dth}.
   For optimality let \(A\) be \(s\)-quasiseparable given in \(t\)-\HSS form. We use \cref{prop:altdef}:
   \begin{equation}
     \label{eq:rkbound}
     \begin{bmatrix} A_{K ; 3\dots4,1\dots 2} & A_{K ; 3\dots4,5\dots 6} \end{bmatrix} =
       \begin{bmatrix}
    U_{K ; 3} R_{K ; 3}\\ U_{K ; 4} R_{K ; 4} \end{bmatrix} H
     \end{equation}
   where \(H \in K^{t \times 4t}\). The quasi-separability of \(A\) bounds the rank of the left part of
   \cref{eq:rkbound} by \(2s\) while the one of the right side is bounded by \(t\).
     When the first bound is tight we get \(t\geq 2s\).
 \end{proof}
 

\subsection{\Bruhat generators}

The \Bruhat generator was first defined in \cite{P16,PS18}.
Contrarily to \SSS and \HSS, it does not use on a pre-defined grid but relies on the rank profile
information contained in the rank profile matrix~\cite{DPS17} of the lower and upper triangular
parts of the quasiseparable matrix.

Recall from \cite{PS18} that a matrix is \(t\)-overlapping if any subset of \(t + 1\) of its non-zero columns
(resp. rows)
contains at
least one whose leading non-zero element
is below (resp. before) the trailing non-zero element of another.
We call \(\J n\) the anti-identity matrix of dimension \(n\) and define the {\it Left} operator
\(\ultriangle:\K^{n\times n} \rightarrow \K^{n\times n}\) s.t.
\begin{equation}
  \lft{A}_{i,j} = \left\{\begin{matrix} A_{i,j} & \text{ if } i + j \leq n \\
  0 & \text{ otherwise } \end{matrix} \right..
  \end{equation}
\begin{definition}
  An \(n\times n\) matrix \(A\) is represented in \(t\)-\Bruhat format if it is given by a diagonal matrix
  \(D\in \K^{n\times n}\) and 6 matrices \(
  \Tag{C}{L},
  \Tag{R}{L},
  \Tag{E}{L},
  \Tag{C}{U},
  \Tag{R}{U},
  \Tag{E}{U}\)
  where \(\Tag{C}{L}\in\K^{n \times u}\) and \(\Tag{C}{U}\in\K^{n \times v} \) are
  in column echelon form and \(t\)-overlapping, \(\Tag{E}{L}\in\K^{u\times n}\) and
  \(\Tag{E}{U}\in\K^{v\times n}\) are  in column echelon form and \(t\)-overlapping
  and \(  \Tag{R}{L}\in\K^{u \times u},   \Tag{R}{U}\in\K^{v\times v}\) are permutation matrices and satisfy
  \[
A= D+  \J{n}\lft{\Tag{C}{L}\Tag{R}{L}\Tag{E}{L}} + \lft{\Tag{C}{U}\Tag{R}{U}\Tag{E}{U}} \J{n}
  \]
\end{definition}

\begin{proposition}
  Any \(n \times n\) \(s\)-quasiseparable matrix has an \(s\)-\Bruhat representation. It uses
  \(\Space{Bruhat}(n,s)=4ns\) field elements which is optimal.
\end{proposition}
\begin{proof}
  By~\cite[Theorem~20]{PS18}. As \(2ns\) coefficients are necessary  to represent all rank \(s\)
  triangular  matrices, \(4ns\) is optimal.
\end{proof}
\section{Construction of the generators}
\subsection{\SSS generator from a dense matrix}
We recall in \algoName{algo:DenseToSSS} the construction of an \SSS  generetor from a dense
\(s\)-quasiseparable matrix \(A\in\K^{n\times n}\).
It is adapted from \cite[\S 6.1]{CDGPV03} and \cite[Alg. 6.5]{EG05bsss} where the
SVD based numerical rank revealing factorizations are replaced by~\PLUQ. 


The blocks \(D_i\) are directly extracted from the dense matrix in \cref{step:di}.
Each block-triangular part is then compressed independently.
Each step eliminates a chunk made of a block-row of A and a remainder from the previous step.
The result is three blocks of the generator and a remainder to be eliminated at the subsequent step.

\begin{algorithm}[htb]
   \algoCaptionLabel{DenseToSSS}{}
   \begin{algorithmic}[1]
     \algrenewcommand{\algorithmicensure}{\emph{InOut:}} 
     \Require{$A$ an $n\times n$ $s$-quasi-separable matrix with \(s \leq t\)}
     \Ensure{\(P_i, Q_i, R_i, U_i, V_i, W_i, D_i\) for appropriate \(i \in \intset{1,N}\)
  a \(t\)-\SSS representation of \(A\)}
     \State \(A = \begin{smatrix} A_{1, 1} & \cdots & A_{1, N} \\
       \vdots & & \vdots \\
       A_{N, 1} & \cdots & A_{N, N} \end{smatrix}, H  = \begin{smatrix} H_{0, 1} & \cdots & H_{0, N} \\
       \vdots & & \vdots \\
       H_{N, 1} & \cdots & H_{N, N} \end{smatrix}\gets 0\)
       \For {\(k = 1\dots N - 1\)}
       \State \(D_k \gets A_{k,k}\) \label{step:di}
       \State \(\paren{\begin{bmatrix} W_k \\ U_k\end{bmatrix}, \begin{bmatrix} V_{k + 1} &
             H_{k, k + 2 \dots N} \end{bmatrix}} \gets
       \PLUQ\paren{\begin{bmatrix} H_{k - 1, k + 1 \dots N}  \label{step:pluq}\\
         A_{k, k + 1 \dots N}  \end{bmatrix}}\) 
       \State \(\paren{\begin{bmatrix} Q_{k + 1} \\ H_{k + 2 \dots N, k }  \end{bmatrix}, \begin{bmatrix} R_k & P_k\end{bmatrix}
       }
       \gets
       \PLUQ\paren{
         \left[\begin{smallmatrix}
           H_{k + 1 \dots N, k - 1} & A_{k + 1 \dots N, k}
       \end{smallmatrix}\right]}\) \label{step:pluq2}
       \EndFor
       \State{\(D_N = A_{N,N}\)}
   \end{algorithmic}
\end{algorithm}
\begin{proposition}
  \label{prop:dts}
  \algoName{algo:DenseToSSS} computes a \(t\)-\SSS generator for an \(s\)-quasiseparable matrix (\(s \leq t\)) in
  \(\Time{DenseToSSS}(n,t)=2C_\texttt{RF}n^2s^{\omega-2}\) field operations.
  \end{proposition}
\begin{proof}
  For \(k \in \intset{1,{N - 1}}\), the dimensions of the output of Lines \ref{step:pluq} and \ref{step:pluq2}
  are sufficient since the input of the factorisation is a concatenation of a block of \(A\) with a
  rank-revealing factor
  of another block of \(A\) on the same side of the diagonal, and is hence of rank at most \(s\).
  
    Let \(i,j \in \intset{1,N}\).
    If \(i = j\) \cref{step:di}
    for \(k = i\)
    gives \(D_i = A_{i,i}\).
    If \(i < j\), \cref{step:pluq} gives 
    \begin{align}
      W_{j - 1}V_j &= H_{j-2,j} \\
      W_kH_{k,j} &= H_{k - 1, j} & ( k &\in \intset{{1},{j - 2}})\\
      U_i H_{i,j} &= A_{i, j} & (i &< N) \\
      U_iV_{i + 1} &= A_{i, i+1}
    \end{align}
    which combines to \(U_i W_{i + 1} \dots W_{j - 1} V_j = A_{i, j}\).
    The same way, if \(i > j\) then \(P_i  R_{i - 1} \dots R_{j + 1} Q_j = A_{i,j}\).

    The cost is
    \(
    \sum_{k = 1}^{N - 1} 2 \Time{RF}(t(N-k), 2t, s) = 2C_\texttt{RF} n^2 s^{\omega-2}
    \).
    \end{proof}

\subsection{\HSS generator from a dense matrix}

The first construction algorithm for a general quasiseparable matrix is presented in \cite{CGP06ULV}.
We present in \algoName{algo:DenseToHSS} an iterative version of the faster and simpler algorithm of \cite{XCGL10HSS}.

Each step of the loop on \(k\) passes block-row-wise and block-column-wise on the matrix inherited from the previous step,
factorising block rows and block columns two by two.
At each step each block is hence factorised twice, producing transition matrices \(R\) and \(W\), the remainder being either
passed to the following step or finally stored as a \(B\) matrix.

\begin{algorithm}[htbp]
   \algoCaptionLabel{DenseToHSS}{}
   \begin{algorithmic}[1]
     \algrenewcommand{\algorithmicensure}{\emph{InOut:}} 
     \Require{$A$ an $n\times n$ quasiseparable matrix of order s}
     \Ensure{\(U_{K;i}, V_{K;i}, D_i, R_{k;i}, W_{k;i}, B_{k;i}\) for appropriate \(k \leq K, i \leq 2^k\)
  a \(t\)-\HSS representation of \(A\) with \(t \geq 2s\)}
     \State \(H \gets A\) \label{step:ha}
\Comment{ Use the block division of \cref{eq:kdiv} with \(k = K\)}
     \For{\(i = 1 \dots 2^K\)}
     \State \(D_i \gets A_{K;i,i}\) \label{step:hdi}
     \EndFor
         \For {\(k = K  \dots 1\)}
         \For{\(i = 1 \dots 2^{k }\)} \Comment{All operations are in this loop}
         \Statex \CommentLine{\(R_{K + 1 ; 2 i}\)  (resp. \(W_{K + 1 ; 2 i}\) )
has row (resp. column) dimension 0}
\State \(\paren{
  \begin{bmatrix} R_{k + 1; 2 i - 1} \\ R_{k + 1 ; 2 i} \end{bmatrix},
  \begin{bmatrix} \hlu_{k ; i, 1 \dots i - 1}  &
\hu_{k;i , i + 1 \dots 2^k}\end{bmatrix}
} \gets
\PLUQ\paren{
\left[\begin{smallmatrix} \hl_{k; i, 1 \dots i - 1} &
H_{k; i, i + 1 \dots2^k}\end{smallmatrix}\right]}\) \label{step:pluq1}
\State \(\paren{\left[\begin{smallmatrix}\hlu_{k;1 \dots i - 1, i} \\ 
\hl_{k; i + 1 \dots 2^k, i }\end{smallmatrix}\right],
         \left[\begin{smallmatrix} W_{k + 1; 2 i - 1} & W_{k + 1 ; 2 i} \end{smallmatrix}\right]} \gets
\PLUQ\paren{\left[\begin{smallmatrix} \hu_{k;1 \dots i - 1,  i } \\
H_{k; i + 1 \dots 2^k,  i }\end{smallmatrix}\right]}\) \label{step:plc}
\EndFor
     \For{\(i = 1 \dots 2^{k - 1}\)} \Comment{Only renaming from here}
     \State \(B_{k; 2 i - 1} \gets \hlu_{k ; 2 i - 1, 2 i}\)
     \State \(B_{k; 2 i} \gets \hlu_{k ; 2 i , 2 i - 1}\)
     \For{\(j = 1 \dots 2^{k - 1}, j \neq i\)}
     \State \(H_{k - 1 ; i, j} \gets \begin{bmatrix} \hlu_{k ; 2 i - 1, 2 j - 1} & \hlu_{k ; 2 i - 1, 2 j} \\
       \hlu_{k ; 2 i  , 2 j - 1} & \hlu_{k ; 2 i , 2 j} \end{bmatrix}\) \label{step:rec}
     \State \(    H_{k - 1 : j, i} \gets \begin{bmatrix} \hlu_{k : 2 j - 1, 2 i - 1} & \hlu_{k : 2 j - 1, 2 i} \\
        \hlu_{k : 2 j  , 2 i - 1} & \hlu_{k : 2 j , 2 i} \end{bmatrix}\) 
     \EndFor
     \EndFor
     \EndFor
     \For{\(i = 1 \dots 2^K\)}
     \State \( U_{K ; i} \gets R_{K + 1 ; 2 i - 1}\) \label{step:u}
     \State \( V_{K ; i} \gets W_{K + 1 ; 2 i - 1}\)\label{step:v}
         \EndFor
   \end{algorithmic}
\end{algorithm}

\begin{proposition}
  \label{prop:dth}
  \algoName{algo:DenseToHSS} computes a \(t\)-\HSS generator for an \(s\)-quasiseparable matrix if \(2s \leq t\) in
  \(C_\texttt{RF}n^2t^{\omega-2}\) field operations. For \(t=2s\),
  this is \(\Time{DenseToHSS}(n,s)=2^{\omega}C_\texttt{RF}n^2s^{\omega-2}\).
\end{proposition}

\begin{proof}
  Let \(k \in \intset{1,K}\), \(i\neq j \in \intset{1,{2^k}}\).
  The dimensions of the output in Lines \ref{step:pluq1} and \ref{step:plc} is sufficient since the matrices being factorized are each time a concatenation of two blocks of rank at most \(s\) and are hence of rank at most
  \(2s \leq t\).
  If \(|i - j| = 1\), 
the instructions give
  \begin{equation}
    \label{eq:iej}
    H_{k;i, j} = \begin{bmatrix}R_{k + 1 ; 2i - 1} \\ R_{k + 1 ; 2i } \end{bmatrix} B_{k ; i} \begin{bmatrix}W_{k + 1 ; 2j - 1} & W_{k + 1 ; 2j } \end{bmatrix}.
  \end{equation}
  Otherwise,
  \begin{equation}
    \label{eq:inj}
    H_{k;i, j} = \begin{bmatrix}R_{k + 1 ; 2i - 1} \\ R_{k + 1 ; 2i } \end{bmatrix} \hlu_{k ; i, j} \begin{bmatrix}W_{k + 1 ; 2j - 1} & W_{k + 1 ; 2j } \end{bmatrix}.
    \end{equation}
    Let now \(I,J \in \intset{1,N}\).
    If \(I = J\), \cref{step:hdi}
    gives \(D_I = A_{K;I,J}\).
    Otherwise, let \(k\) be the highest level of recursion for which \(A_{K;I,J}\) is not included in a diagonal
    block.
    From \cref{step:ha}, \(A_{K;I,J} = H_{K;I,J}\).
    \Cref{eq:inj} can be used \(K - k\) times, together with \cref{step:rec} to get
    \begin{equation}
      \label{eq:prod}
      A_{K;I,J} = 
      R_{K + 1; 2 I - 1} \dots R_{k + 2; i_2}\hlu_{k + 1; i_1, j_1}
      W_{k + 2 ; j_2} \dots W_{K + 1 ; 2 J - 1}
    \end{equation}
    where \(i_2 = \lfloor I / 2^{K - k - 2}\rfloor, i_1 = \lfloor I / 2^{K - k - 1}\rfloor, j_1 =
    \lfloor J / 2^{K - k - 1}\rfloor\) and  \(j_2 = \lfloor J/ 2^{K - k - 2}\rfloor\).
    \(R_{K + 1; 2 I - 1}\), \(W_{K + 1 ; 2 J - 1}\) and \(\hlu_{k + 1; i_1, j_1}\) can be replaced in \cref{eq:prod} using Lines
    \ref{step:u} and \ref{step:v} and \cref{eq:iej} (from the definition of \(k\) we have
    \(|i_1 - j_1| = 1\))
    in order to get \cref{eq:altdef};
    this concludes the proof of correctness.

    \Cref{step:pluq1} at \(k<K\) and \(i\) peforms a rank revealing decompositions on an input formed
    by the \(2t\times (i-1)t\) block \(\hu_{k;i,1 \dots i - 1 } \) and
    the \(2t\times 2t(2^k-i)\) block \(H_{k; i, i + 1 \dots 2^k }\) at cost \(\Time{RF}(t(2^{k+1} - i), 2t, t)\).
    The cost is equal for~\cref{step:plc}.
    The overall cost is then
    \(
    \sum_{k=1}^{\log_2 \frac{n}{t}} \sum_{i=1}^{2^k} 4C_\texttt{RF}(2^{k+1}-i)t^\omega
    \leq 4C_\texttt{RF}n^2t^{\omega-2}
    \leq 2^{\omega}C_\texttt{RF}n^2s^{\omega-2}.
    \)
\end{proof}

Because
the blocks of each side of the diagonal are defined by the same matrices,
\algoName{algo:DenseToHSS}
and any \HSS construction
algorithm
applies rank revealing
factorisations on blocks with rank bounded by \(2s\) for
\(s\)-quasiseparable matrices instead of
\(s\) in \algoName{algo:DenseToSSS}. 
The optimal \HSS block size of \(s\)-quasiseparable matrices is thus \(2s\),
which makes \HSS less efficient
in terms of storage and operation cost.

As  the costs are higher and \HSS has the same drawbacks as \SSS, namely needing a fixed slicing grid and a
previously computed quasiseparability order,
we do not detail more algorithms for \HSS.
For information in the numerical context we mainly refer to \cite{lyons2005fast, SDC07}.
Note that faster construction algorithms exist, probabilistic in \cite{Mart11HSS} and with constraints on the input in
\cite{CGP06ULV}.

\subsection{Bruhat generator from a dense matrix}
\label{sec:dtb}
The construction of a \Bruhat generator from a dense matrix is achieved by
\cite[Alg.~12]{PS18}
run twice, once for each of the upper and lower triangular parts of the input matrix,
and the diagonal matrix \(D\) is directly extracted from the dense matrix.

We give in \algoName{algo:LBruhatGen} an updated version of \cite[Alg.~12]{PS18},
where Schur complement computations are delayed until they are needed.
This allows for faster computations when the input is not given as a dense matrix
and will be used for computing the sum of two matrices in \Bruhat form and generators from a sparse matrix.

\algoName{algo:LBruhatGen} can be given any input format,
provided we have a way to compute for any submatrix \(B\) of the input matrix
\begin{enumerate}
\item \(\bbcre(B, G, H)\) a CRE decomposition of \(B - G \trsp H\); \label{it:bbcre}
\item for \RRP a set of indices,
  \(\exprows B{\RRP}\) and \expcols B \RRP the rows and columns of \({B}\)
  with indices in \RRP. \label{it:expand}
  \end{enumerate}

We use the notation 
\texttt{TRSM} for {\it TRiangular Solve Matrix}: \texttt{TRSM}\((L,A)\) outputs \(\inv L A\) for \(L\) triangular.
\begin{algorithm}[htbp]
  \algoCaptionLabel{LBruhatGen}{}
    \begin{algorithmic}[1]
      \Require{\(A\in\K^{m\times m}\) left triangular \(s_A\)-quasiseparable}
      \Require{\(G,H \in \K^{m\times t}\)}
      \Comment{\(t=0\) on the first call}
      \Ensure{\(C, R, E\) a left-\Bruhat generator for \(A - \lft{G\trsp H}\)}

      \State Split
      \(
A =      \begin{bmatrix}\Tag{A}{11}&\Tag{A}{12}\\\Tag{A}{21}&\end{bmatrix},
G =        \begin{bmatrix}\Tag{G}{1} \\\Tag{G}{2}  \end{bmatrix},
H =        \begin{bmatrix} \Tag{H}{1}\\ \Tag{H}{2}  \end{bmatrix}
\) where \(\Tag{A}{11}\in \K^{\frac{m}{2}\times \frac{m}{2}}\)
\State \(C_0,R_0,E_0  \assign \bbcre(\Tag{A}{11},\Tag{G}{1},\Tag{H}1)\) \label{step:bbcre}
\State \(\RRP \assign \text{\texttt{RRP}}(C_0); \CRP \assign \text{\texttt{CRP}}(E_0);
r_0\assign \#\RRP \) \label{step:rp}
\State  \( \begin{bmatrix} U & V \end{bmatrix} \assign  E_0 Q_{\mathcal C}\)
where \(U\in \K^{r_0\times r_0}\) is upper triangular.
\State \(\begin{bmatrix} L \\ M \end{bmatrix} \assign P_{\mathcal R} C_0 \)
where \(L \in \K^{r_0\times r_0}\) is lower triangular.
\State \(X\assign \exprows{\Tag{A}{12}}{\RRP} - \rows{\Tag{G}{1}}{\RRP}\trsp {\Tag{H}{2}}\) \label{step:cbx}
\State \({\Tag{B}{12}}\!\assign \lftr{\trsm \paren{L, X}}\) \label{step:tr1}
\Comment{\(\rows{\Tag{A}{12}}{\RRP} = \lftr{L \Tag{B}{12}
    + \rows{\Tag{G}{1}}{\RRP}\trsp{\Tag{H}{2}}}\)}
\State \(Y\assign \expcols{\Tag{A}{21}}{\CRP} -  \Tag{G}{2}\trsp{\rows{\Tag{H}{1}}{\CRP}}\) \label{step:cby}
\State \(\trsp{\Tag{B}{21}}\assign \lftc{\trsm(\trsp U,\trsp{Y})}\) \label{step:tr2}
\Statex
\Comment{\(\cols{\Tag{A}{21}}{\CRP} = \lftc{ \Tag{B}{21} U
    + \Tag{G}{2}\cols{\trsp{\Tag{H}{1}}}{\CRP}}\)}
\State{\(C_1, R_1, E_1 \assign\)}
  \Statex\(\ltcre
    \left(\expcols{\Tag{A}{21}}{\compl{\CRP}}\rows{\Id}{\compl{\CRP}},
    \begin{smatrix}\Tag{G}2 & \Tag{B}{21}\end{smatrix},
    \cols{\Id}{\compl{\CRP}}
    \begin{smatrix} {\rows{\Tag{H}{1}}{\compl{\CRP}}} & \trsp{V}  \end{smatrix}
    \right)\) \label{step:cbc1}
  \State\(C_2, R_2, E_2 \assign\)
    \Statex\(\ltcre
  \left(\cols{\Id}{\compl{\RRP}}\exprows{\Tag{A}{12}}{\compl{\RRP}},
  \cols{\Id}{\compl{\RRP}}
  \begin{smatrix}\rows{\Tag{G}{1}}{\compl{\RRP}} & M\end{smatrix},
    \begin{smatrix} {\Tag{H}{2}} & \trsp{\Tag{B}{12}}  \end{smatrix}\right)\)  \label{step:cbc2}
\State{\(P_{01} \assign\) the permutation which sorts the rows of \(E_0\) and \(E_1\) by
  increasing column of pivot}
\State{\(P_{02} \assign\) the permutation which sorts the columns of \(C_0\) and \(C_2\)
  by increasing row of pivot}
\State{\(C \assign \begin{smatrix} C_0 & C_2  & \\
    {\Tag{B}{21}}\trsp{R_0} & & C_1\end{smatrix}\begin{smatrix}P_{02} & \\
      & I\end{smatrix}\)}
\State{\(R \assign \begin{smatrix} \trsp{P_{02}}& \\
    & I \end{smatrix}\begin{smatrix} {R_{0}}& & \\
    & & R_2 \\ & R_1 & \end{smatrix}\begin{smatrix} \trsp{P_{01}}& \\ & I \end{smatrix}\)}
\State{\(E \assign \begin{smatrix}P_{01} & \\
    & I\end{smatrix}\begin{smatrix} E_0 & \trsp{R_0} {\Tag{B}{12}} \\
      E_1 & \\ & E_2\end{smatrix}\)}
        \State \Return \(C, R, E\)
  \end{algorithmic}
\end{algorithm}
\begin{proposition}
  \label{prop:dtb}
  An \(s\)-\Bruhat generator can be computed from an \(n \times n\) dense \(s\)-quasiseparable matrix in 
  \(\Time{DenseToB}(n,s) = C_\texttt{RF} n^2s^{\omega - 2}\).
\end{proposition}
\begin{proof}
  \algoName{algo:LBruhatGen} is adapted from~\cite[Alg.~12]{PS18}; we therefore refer to the
  proof of~\cite[Theorem~24]{PS18} for its correctness.
  Apart from the order in which they are made, the operations are the same in both algorithms when the input is dense and the cost is hence the same.
Computing a \Bruhat generator from a dense matrix is two applications of \algoName{algo:LBruhatGen}.
The cost satisfies:

\(\Time{LBG}(n,s) \leq C_\texttt{RF}/4 n^2s^{\omega - 2}+2\Time{LBG}(n/2,s)  
\leq C_\texttt{RF}/2 n^2 s^{\omega - 2}.
\)
\end{proof}

\subsection{Bruhat generator from a sparse matrix}
In applications, matrices are often presented in a sparse structure. In order to detect and/or harness
their quasiseparable structure, it is crucial to exploit the sparsity in the construction of the
quasiseparable generators.




For the construction of a \Bruhat generator, the generic algorithm \algoName{algo:LBruhatGen} can be
applied on a sparse matrix, provided two operations are specialized:
\begin{enumerate}
\item the extraction of a subset of \(\leq s\) rows or columns into a dense format, which is
  straightforward for a sparse matrix;
\item the computation of a CRE decompoistion, which is specialized in~\algoName{algo:SparseCRE}
  which in turn uses \algoName{algo:SparseRankProfiles}
\end{enumerate}
\begin{algorithm}[htb]
  \algoCaptionLabel{SparseCRE}{}
    \begin{algorithmic}[1]
      \Require{\(A\in \K^{m \times m}\) a rank \(\leq s\) sparse matrix}
      \Require{\( G, H\in \K^{m\times t}\)}
      \Ensure \(C,R,E\) such that \(A = CRE + G\trsp H\)
\State \(\RRP, \CRP \assign \Call{algo:SparseRankProfiles}{A, G, H} \) 
\State \(P=\begin{bmatrix} \rows{\Id}{\RRP}  \\ \rows{\Id}{\compl{\RRP}}  \end{bmatrix} \) ; 
   \(Q =   \begin{bmatrix}  \cols{\Id}{\CRP} &  \cols{\Id}{\compl{\CRP}}  \end{bmatrix}\) \label{alg:csu:q}
   \Statex
   \CommentLine{With {\(\Tag{\bar A}{11} \in \K^{|\RRP|\times|\RRP|}\)} write
     \(P\paren{A - G\trsp H}Q =
\begin{bmatrix}   \Tag{\bar A}{11}&  \Tag{\bar A}{12}\\ \Tag{\bar A}{21}& \Tag{\bar A}{22}\end{bmatrix} -
\begin{bmatrix}   \Tag{\bar G}{1} \\  \Tag{\bar G}{2}\end{bmatrix}
\trsp{\begin{bmatrix} \Tag{\bar H}{1}\\ \Tag{\bar H}{2}\end{bmatrix}}\)}
\State \(\Tag{M}{11} \assign \Tag{\bar A}{11}- \Tag{\bar G}{1} \trsp{(\Tag{\bar H}{1})}\) \label{alg:csu:a11}
\State \(\Tag{M}{12} \assign \Tag{\bar A}{12}- \Tag{\bar G}{1} \trsp{(\Tag{\bar H}{2})}\) \label{alg:csu:a11}
\State \(\Tag{M}{21} \assign \Tag{\bar A}{21}- \Tag{\bar G}{2} \trsp{(\Tag{\bar H}{1})}\) \label{alg:csu:a11}

\State \((L,R,U) \assign \cre\paren{\Tag{M}{11}}\)  \label{alg:CREsparse:cre1}
\State \(C\assign \texttt{TRSM}(L, \Tag{\bar M}{12})\) \Comment{\(C= L^{-1} (\Tag{\bar A}{12}-\Tag{G}{1}\Tag{H}{2})\)} \label{alg:CREsparse:ostrsm}
\State \(D\assign \texttt{TRSM}(\Tag{\bar A}{21},\trsp{U})\)\Comment{\(D=(\Tag{\bar A}{21}-\Tag{G}{2}\Tag{H}{1})U^{-1}\)}
\State \(E\assign \begin{bmatrix}  U & \trsp{R}C\end{bmatrix} \trsp{Q}\)
  \State \(C\assign \trsp{P} \begin{bmatrix} L  \\ D\trsp{R}\end{bmatrix} \)
    \State \Return \((C,R,E)\)
  \end{algorithmic}
\end{algorithm}

\begin{algorithm}[htb]
  \algoCaptionLabel{SparseRankProfiles}{}
    \begin{algorithmic}[1]
      \Require{\(A\in\K^{n\times n}\) a sparse matrix of rank $\leq s$.}
      \Require{\(G,H \in \K^{n\times t}\) dense matrices}
      \Ensure{\(\RRP_A,\CRP_A\) the row and column rank profiles of \(A-G\trsp H\)}
\State \(\Tag{T}{1} \assign \) a unif. random \(n\times (s+t)\) Toeplitz matrix from \(S\subseteq \K\)
\State \(\Tag{T}{2} \assign \) a unif. random \( (s+t) \times n\) Toeplitz matrix from \(S\subseteq \K\)
\State \(K \assign  \trsp{H}  \Tag{T}{1}\)  \label{line:toep1}
\State \(L \assign \Tag{T}{2} G\) \label{line:toep2}
\State \(P \assign A \Tag{T}{1} -  GK\)
\State \(Q \assign \Tag{T}{2} A -  L\trsp{H}\)
\State \Return \texttt{RowRankProfile}\((P)\), \texttt{ColRankProfile}\((Q)\)
  \end{algorithmic}
\end{algorithm}

\begin{lemma}
  \label{lem:sprp}
  \algoName{algo:SparseRankProfiles} is correct with probablity at least \(1-2r/|S|\) and runs in
  \(\Time{SparseRP}(n,r)=2(C_\omega +C_\texttt{RF})nr^{\omega-1}+2r|A|\) with \(r=t+s\).
\end{lemma}
\begin{proof}
Applying the Toeplitz pre-conditionners in~\cref{line:toep1,line:toep2} costs
\(\frac{nt}{r}\softO{r}\) which is domintated by \(nr^{\omega-1}\).
\end{proof}

\begin{proposition}
  \label{prop:spc}
  \algoName{algo:SparseCRE} computes a CRE decomposition of \(A - G \trsp H\) with probablity at least \(1-2r/|S|\) in
  \(\Time{SparseCRE}(n,r) =
  \left(\frac{2^{\omega}-3}{2^{\omega-2}-1}C_\omega+2C_\texttt{RF}\right) n r^{\omega - 1}+2r|A|\) field operations for
  \(s + t \leq r\).
  \end{proposition}
\begin{proof}
    Let \(\rho\) be the rank of \(A-G\trsp H\).
    \begin{equation*}
      \begin{split}
    \Time{SparseCRE}(n,r)&=2\Time{MM}(n,r,t) + \Time{CRE}(\rho,\rho,\rho) + 2\Time{TRSM}(n-\rho,\rho)  \\ & +\Time{SparseRP}(n,r)\\
     & \leq  nr^{\omega-1}\left(4C_\omega+\frac{2C_\omega}{2^{\omega-1}-2}+2C_\texttt{RF}\right)+2r|A|.
      \end{split}
      \end{equation*}
\end{proof}

\begin{proposition}
  \algoName{algo:LBruhatGen} computes a Left-\Bruhat generator from an sparse \(s\)-quasiseparable matrix \(A \in \K^{n\times n}\) in
  \begin{equation*}
    \Time{SpGenB}(n,s, |A|)=
    \paren{\frac{2^{\omega + 1}-9}{2^{\omega-1}-2} \cw + \crf} n s^{\omega - 1} \log n/s + 2s|A|
  \end{equation*}
  field operations with probability at least \(1-2n/|S|\).
\end{proposition}

\begin{proof}
  First, remark that the \(G\) and \(H\) matrices correspond to delayed Schur complement updates for
pivots processed in the previous calls. Hence, in every call to \algoName{algo:LBruhatGen}, these
pivots are located to the left, to the top or in the left-top corner of the work matrix. The
quasiseparable condition imposes that there are  \(t\leq 2s\) of them.
Moreover, in the call to \algoName{algo:SparseCRE}, the ranks verify \(r_A+r_B+t \leq s\). Hence
we can bound \(t\) and write the cost of \algoName{algo:LBruhatGen} only in terms of \(n\) the dimension of
the matrix, \(s\) the initial quasiseparability order, and \(|\cdot|\) the amount of non-zero coefficients of the submatrices we consider.
\begin{align*}
      T(n,s, |A|) &\leq &&T(n/2,s, |A_2|)+ T(n/2, s, |A_3|) \\ &&&+
      \Time{SparseCRE}(n/2,s, |A_1|) \\
      &&&+ 
      2\Time{MM}(s, 2s, n / 2) + 2\Time{TRSM}(s, n/ 2)\\
      &\leq &&T(n/2,s, |A_2|)+ T(n/2, s, |A_3|) +
    2s|A_1| \\
    &&&+    \paren{\frac{2^{\omega + 1}-9}{2^{\omega-1}-2} \cw + \crf} n s^{\omega - 1}.
\end{align*}
   The failure probability is obtained by a union bound on the failure probability of each of the  \(n/s\) calls to \algoName{algo:SparseCRE}.
\end{proof}
We are not aware of any similar algorithm for computing an \SSS or \HSS generator using the sparsity of the input matrix and can hence
only compare our result to the quadratic generation from a dense matrix.
\subsection{Experimental comparison}
\label{sec:ecd}
To complement the asymptotic cost analysis, we present in~\cref{fig:dtq} experiments comparing the
computation time for the construction of \SSS and \Bruhat generators.
The timings for \Bruhat are sub-linear in \(s\), as could be expected from \cref{prop:dtb} but
also slightly depends on \(r\) which comes from neglected costs arising e.g. from the
numerous permutations.
The \SSS cost is constant on our values for reasons we are unable to explain yet.
It is almost always lower than the \Bruhat cost. Yet remember that \algoName{algo:DenseToSSS} takes the quasiseparable order as input, so it has to be computed first (for example with \algoName{algo:LBruhatGen}).

\section{Application to a block vector}

We study here the application of an \(s\)-quasi-separable matrix \(A \in \K^{n \times n}\) given by its generators
(\SSS or \Bruhat) to a block of \(v\) vectors \(B \in \K^{n \times v}\).
We give the costs for \(v \leq s\) (they can be otherwise deduced by slicing \(B\) in blocks
of \(s\) columns).

\subsection{\SSS \(\times\) dense}

We here recall the algorithm of \cite[\S 2]{CDGPV03} for computing the product of
an \SSS matrix with a dense
matrix (independently published in \cite[Alg. 7.1]{EG05bsss}).
For simplicity, \algoName{algo:LowSSSxDense} only details the computations with a strictly
lower-block-triangular
\SSS matrix,
that is a matrix whose \SSS representation is zero except for the \(P_i, Q_i\) and \(R_i\).
Extrapolating from there to the product with any \SSS matrix can be done by transposing
the algorithm for the upper-block-triangular part, and adding the product with the
block-diagonal matrix made of the \(D_i\).
\begin{algorithm}[htbp]
\algoCaptionLabel{LowSSSxDense}{}
   \begin{algorithmic}[1]
     \algrenewcommand{\algorithmicensure}{\emph{InOut:}} 
     \Require{
       \(P_i, Q_i, R_i\) for 
       \(i \in \intset{1,N}\)
an \(s\)-\SSS generator for a strictly lower-block-triangular
 matrix \(A\); 
     \(B\text{ and }C\) dense \(n \times v\) matrices}
     \Ensure{\(C += AB\)}
     \State Split \(B = \begin{bmatrix} B_{1}\\
\vdots \\
       B_{N}  \end{bmatrix}\), \(C = \begin{bmatrix} C_{1}\\
\vdots \\
C_{N}  \end{bmatrix}\) in \(s\times s\) blocks
     \State \(\Hi 1 \gets Q_1 B_1\)
     \For {\(i = 2\dots N\)}
     \State \(\Hi i \gets Q_i B_i + R_i \Hi{i - 1}\)
     \State \(C_i \gets C_i + P_i\Hi{i - 1}\)
\EndFor
   \end{algorithmic}
\end{algorithm}
\begin{proposition}
  \label{prop:sxd}
The product of an \(n \times n\) matrix given by its \(s\)-\SSS generator with an \(n \times v\) dense
matrix with \(v \leq s\) can be computed in 
\(\Time{SxDense}(n,s,v)=  7C_\omega nsv^{\omega-2}\).
\end{proposition}
\begin{proof}
In \algoName{algo:LowSSSxDense} we have by induction that
  \begin{equation}
    \text{for }i \in \intset{1,N}, \Hi i =  \sum_{j = 1}^{i } R_{i} \dots R_{j + 1} Q_jB_j.
  \end{equation}
As the blocks of the product follow
\begin{align}
  C_i 
  &= P_i \sum_{j = 1}^{i - 1} R_{i - 1} \dots R_{j + 1} Q_jB_j, \label{eq:ci}
\end{align}
\(\Hi{i - 1}\) can be multiplied once by \(P_i\) to compute \(C_i\) and once by
\(R_i\) to compute the following blocks.
  The cost is \(N \times C_\omega s^2v^{\omega-2}\) for the diagonal blocks and two applications of
  \algoName{algo:LowSSSxDense} in which 
  each step  costs \(3C_\omega s^2v^{\omega-2}\).
  \end{proof}
\subsection{\Bruhat \(\times\) dense}

\begin{proposition}
  \label{prop:bxd}
  The product of an \(n\times n\) matrix given by its \(s\)-\Bruhat generator by a
  dense \(n\times v\) matrix with \(v\leq s\) can be computed in
  \(
  \Time{BxDense}(n,s,v)=8C_\omega nsv^{\omega-2}
  \).
\end{proposition}
\begin{proof}
  This is given 
  by \cite[Alg.~14]{PS18} called twice on the lower and upper triangular part
of the quasiseparable matrix.
\end{proof}

Note that in order to benefit from fast matrix multiplication, the \Bruhat generator (using
\(4ns\) space) needs to be transfered into a Compact-\Bruhat form, by storing each echelon
from into two block diagonal matrices using twice as many field elements (additonal ones being
zeros). This compression can be done online, hence the space storage remains \(4ns\), but the cost
of the product by a dense matrix becomes \(8C_\omega nst^{\omega-2}\) hence losing the advantage over the \SSS
format (with cost \(7C_\omega nsv^{\omega-2}\) for the same operation).

\subsection{Experimental comparison}
\label{sec:eca}
Experimental results are given in \cref{fig:apq} (\cref{app:fig}).
As expected from Propositions \ref{prop:sxd} and \ref{prop:bxd} we obtain costs that are linear in \(s\); we can also observe the same slight dependance in \(r\) of the \Bruhat cost as in \cref{sec:ecd}.
On the parameters we chose, \SSS is about four times faster than \Bruhat.
This can be explained by the compactification of the \Bruhat generator needed for the product.
This operation is free of arithmetic operations and hence does not appear in the cost of
\cref{prop:bxd} but the data tranfers are non-negligible in practice.
\section{Sum of quasiseparable matrices}
The sum and product of two quasiseparable matrices of order \(s_B\) and \(s_C\) 
are quasiseparable matrices of order at most \(s_B + s_C\).
In this section we show how to compute \SSS and \Bruhat generators for the sum of two quasiseparable matrices.

The result we give in \cref{prop:sumcost} for the sum of matrices given in \SSS form can only be used on two
generators defined on the same grid.
This is a drawback of most operations in \SSS which is avoided with the \Bruhat format.
As a consequence, in a large sequence of operations, the \SSS grid size needs to be chosen according
to the maximal quasi-separability order among all intermediate results,
while the \Bruhat always fits to the current quasiseparable order. This can impact the overall  cost.
The slower original \SSS format of~\cite{Eidelman1999OnAN} avoids this issue, at the expense of
multiplying space and time costs by the quasiseparability order, as in~\cite{BEG17,BEG14}.

\subsection{SSS sum}

Consider two matrices \(B\) and \(C\) with the same order \(s\).
We first note that the concatenation of
the blocks of both input generators leads to matrices which satisfy \cref{eq:def}
for \(A = C + B\) \cite[\S 10.2]{CDGPV03}.


Let \(\Tag PK_i, \Tag VK_i
  , \Tag QK_i, \Tag UK_i
  , \Tag RK_i, \Tag WK_i
  , \Tag DK_i
  \) for appropriate \(i \in \intset{1,N}\) be an \(s\)-\SSS representation of $K$ for  $K\in\{B,C\}$.
  \label{prop:apb}
  The following matrices satisfy \cref{eq:def} with $A=B+C$, for appropriate \(i \in \intset{1,N}\).
\begin{align}
 P_i &= \begin{bmatrix}
    \Tag PB_{i}  & \Tag PC_{i}
  \end{bmatrix} ,
 Q_i = \begin{bmatrix}
       \Tag QB_{i}\\    \Tag QC_{i}
\end{bmatrix},
R_i = \begin{bmatrix}
    \Tag RB_{i} &  \\
        & \Tag RC_i
          \end{bmatrix}\\
U_i &= \begin{bmatrix}
   \Tag UB_{i} & \Tag UC_{i}
  \end{bmatrix}, 
 V_i = \begin{bmatrix}
    \Tag VB_{i}\\\Tag VC_{i}        
      \end{bmatrix}, 
W_i = \begin{bmatrix}
    \Tag WB_{i} &  \\
 & \Tag WC_i
          \end{bmatrix}\\
D_i &= \Tag DB_i + \Tag DC_i \label{eq:dab}
\end{align}
Such sets of matrices with these dimensions satisfying \cref{eq:def} 
will be called an \((s, 2s)\)-\SSS generator for \(A\).
The granularity of their description remains that of \(s \times s\) blocks, but the dimension of the matrices in the representation is doubled and leads to a suboptimal storage size.
A second step is therefore to use \algoName{algo:SssCompression} to obtain a \(2s\)-\SSS generator for the sum 
and reduce the storage size by \(4s(n-2s)\).
\begin{algorithm}[htbp]
\algoCaptionLabel{SssCompression}{}
   \begin{algorithmic}[1]
     \algrenewcommand{\algorithmicensure}{\emph{InOut:}} 
     \Require{
\(P_i, Q_i, R_i, U_i, V_i, W_i, D_i\) for appropriate \(i \in \intset{1,N}\),
an \((s, 2s)\)-\SSS generator for \(A \in \K^{n \times n}\)} 
     \Ensure{\(P'_i, Q'_i, R'_i, U'_i, V'_i, W'_i, D'_i\) for appropriate \(i \in \intset{1,M}\),
       a \(2s\)-\SSS representation of \(A\) 
       with
  \(M = \lceil N/2 \rceil\)}
     \For {\(i \gets 1\dots M\)}
\State \(  P'_i \gets \begin{bmatrix}
P_{2i-1} \\
P_{2i} R_{2i - 1}
  \end{bmatrix}\)
\State \(    Q'_i \gets \begin{bmatrix}
R_{2i} Q_{2i-1} &   Q_{2i}
    \end{bmatrix}\)
\State \(    R'_i \gets R_{2i} R_{2i - 1} \)
\State \(      U'_i \gets \begin{bmatrix}
U_{2i-1} W_{2i} \\
U_{2i} 
  \end{bmatrix}\)
\State \(    V'_i \gets \begin{bmatrix}
V_{2i-1} & W_{2i - 1}   V_{2i}
    \end{bmatrix}\)
\State \(    W'_i \gets W_{2i - 1} W_{2i} \)
\State \(    D'_i \gets \begin{bmatrix}
      D_{2i - 1} & U_{2i - 1} V_{2i} \\
      P_{2i} Q_{2i - 1} & D_{2i}
      \end{bmatrix}\)
\EndFor
   \end{algorithmic}
\end{algorithm}
\begin{proposition}
  \label{prop:sumcost}
  A \(2s\)-\SSS representation of \(B + C \in \K^{n \times n}\) can be computed from \(s\)-\SSS
  representations of \(B\) and \(C\) in time
  \begin{equation}
  \Time{S+S}(n,s) \gets \paren{10 + 2^\omega} C_\omega n s^{\omega - 1}.
  \end{equation}
\end{proposition}
\begin{proof}
  For any \(s \times s\) block \(A_{i,j}\) of \(A = B + C\),
  it can be checked that the representation in the output of \algoName{algo:SssCompression} called on the generator of \cref{prop:apb} matches.
  The additions of \cref{eq:dab} are dominated by the call to
  \algoName{algo:SssCompression} whose cost is of \(M\) steps with four \(2s \times 2s\) by \(2s \times s\) products,
  two \(2s \times 2s\) square products, and two \(s \times 2s\) by \(2s \times s\) products.
\end{proof}

Note that the  \((s, 2s)\)-\SSS generator is intermediate between the \SSS form and the original definition of
quasiseparable matrices given in
\cite{Eidelman1999OnAN}, where the generators are \(s \times s\) matrices but the granularity of the
description is of dimension 1.

\subsection{Bruhat sum}
\label{sec:bsum}
As with \SSS,
the sum of two matrices in \Bruhat form can be computed by first concatenation of both generators, then by retrieving the \Bruhat format in a second step.

Given two left triangular matrices \(A\) and \(B\) given by \Bruhat generators
\(\Tag{C}{A},\Tag{R}{A},\Tag{E}{A},\Tag{C}{B},\Tag{R}{B},\Tag{E}{B}\), their sum indeed writes
\begin{equation}\label{eq:sumbb}
A+B = \lft{
\begin{bmatrix} \Tag{C}{A}&\Tag{C}{B}\end{bmatrix}
\begin{bmatrix} \Tag{R}{A}\\ & \Tag{R}{B}\end{bmatrix}
\begin{bmatrix} \Tag{E}{A} \\ \Tag{E}{B}\end{bmatrix}
}.
\end{equation}
A \Bruhat generator for the right side in \cref{eq:sumbb} can be obtained from a call to 
\algoName{algo:LBruhatGen}, viewed here as a compression algorithm. 
This relies on a specific CRE decomposition 
(\algoName{algo:BruhatSumCRE}),
and on having 
\(\exprows D \RRP\)
for \(D\) 
a submatrix of a sum given as in \cref{eq:sumbb} and \RRP a set of
  row indices
  (\cref{prop:exp}).

\begin{algorithm}[htb]
  \algoCaptionLabel{BruhatSumCRE}{}
    \begin{algorithmic}[1]
      \Require{\(A, B \in\K^{n\times n}\)  of rank $\leq r_A$ and \(\leq r_B\) given by
        generators
        \(\Tag{C}{A},\Tag{R}{A},\Tag{E}{A},\Tag{C}{B},\Tag{R}{B},\Tag{E}{B}\)
        s.t. \(A=\Tag{C}{A}\Tag{R}{A}\Tag{E}{A}\) and \( B=\Tag{C}{B}\Tag{R}{B}\Tag{E}{B}\) which
        are submatrices of \Bruhat generators of matrices comprising \(A\) and \(B\)}
      \Require{\(G,H \in \K^{n\times t}\)}
      \Ensure{\(C, R, E\) such that \(A+B  = C R E + G\trsp H\)}
      \State \(\Tag{C}{R},\Tag{R}{R},\Tag{E}{R} \assign \cre\paren{
      \begin{smatrix}
        \Tag{R}{A}\Tag{E}{A}\\
        \Tag{R}{B}\Tag{E}{B}\\
        -\trsp H        
      \end{smatrix}
      }
      \) \label{step:rh}
      \State \(\Tag{C}{L},\Tag{R}{L},\Tag{E}{L} \assign \cre\paren{
        \begin{smatrix} \Tag{C}{A} &\Tag{C}{B} & G    \end{smatrix}
}\)
      \State \(X \assign \Tag{R}{L}\Tag{E}{L} 
      \Tag{C}{R}\Tag{R}{R}\)
      \State \( \Tag{C}{X},\Tag{R}{X},\Tag{E}{X} \assign \cre(X)\)
      \State \(C\assign \Tag{C}{L}\Tag{C}{X}\) \label{step:c}
      \State \(R\assign \Tag{R}{X}\)
      \State \(E\assign \Tag{E}{X}\Tag{E}{R}\) \label{step:e}
    \end{algorithmic}
\end{algorithm}

\begin{proposition}
  \algoName{algo:BruhatSumCRE} computes a CRE decomposition of \(A+B - G \trsp H\) in
  \(\Time{BSumCRE}(n,r) = \paren{3 C_\omega + 2 C_\texttt{RF}} n r^{\omega - 1}\) for \(r_A+r_B + t \leq r\).
  \end{proposition}
\begin{proof}
  The matrices \(C\) and \(E\) are in column and row echelon form respectively as they are products of
  two echelon forms.
  The cost is that of two dense CRE decompositions of size \(n \times (r_A+r_B + t)\) and products of an
  \(n \times (r_A+r_B+t)\) matrix by two 
  \((r_A+r_B + t) \times ( r_A+r_B+t)\) and one \((r_A+r_B + t) \times n
  \) matrices.
  \end{proof}
\begin{proposition}
  \label{prop:exp}
  For \(D\in\K^{n\times n}\) a submatrix of a
  the left-triangular part of a
  sum  as in \cref{eq:sumbb} and \RRP a set of
  \(s\) row indices,
  \(\exprows D \RRP\) can be computed in \(\Time{SumExp}(n, s) = \cw ns^{\omega - 1}\).
\end{proposition}
\begin{proof}
  There are at most \(s_A\) (resp. \(s_B\)) pivots of \(A\) (resp. \(B\)) impacting \(D\).
  We can thus write \(D = CRE\) with \(C\) made of \(n\) rows and \(s_A + s_B\) columns of \(\begin{bmatrix} \Tag{C}{A}&\Tag{C}{B}\end{bmatrix}\),
  \(R\) a permutation  and \(E\) made of \(n\) columns and \(s_A\) rows of \(\Tag{E}{A}\) and \(s_B\) rows of  \(\Tag{E}{B}\).
  \end{proof}
\begin{proposition}
  \label{prop:cbsum}
  The \Bruhat form of the sum of two \(n \times n\) matrices of quasiseparable order \(s_A\) and \(s_B\) in \Bruhat form can be computed
  in
  \(
    \Time{B+B}(n,s) = \paren{\frac{9 \cdot 2^{\omega - 2} - 8}{2^{\omega - 2} - 1}\cw
      + 2\crf} n s^{\omega - 1} \log n/s \label{eq:comp+}
    \)
    for \(s=s_A+s_B\).
\end{proposition}
\begin{proof}
  Each lower and upper triangular part is converted to a left triangular instance and computed
  independently.
\algoName{algo:LBruhatGen} is then called twice with \(t=0\) on an input matrix
in factorized form as in~\eqref{eq:sumbb}.

  The proof is the same as for \cref{prop:spc} except that in the cost, the \Time{SparseCRE} terms
  are replaced by  \Time{BruhatSumCRE} terms and
the rows and columns of the submatrices are computed at a cost given by \Time{SumExp}.
Then we have 
\begin{align*}
      T(n,s) &\leq &&2T(n/2,s)+\Time{BSumCRE}(n/2,s) + 2\Time{SumExp}(n/2,s,s) \nonumber\\
      &&&+ 
      2\Time{MM}(s, 2s, n / 2) + 2\Time{TRSM}(s, n/ 2)\\
      &\leq &&2T(n/2,s) + \paren{\frac{9 \cdot 2^{\omega - 3} - 4}{2^{\omega - 2} - 1}\cw + \crf}ns^{\omega - 1}
    \end{align*}
    for one call to \algoName{algo:LBruhatGen}.
\end{proof}
\section{Product in \SSS}
The product of two matrices given in \SSS form uses two tricks we have seen previously.
The first one is to start by computing an \((s, 2s)\)-\SSS representation before
compression,
as in the sum.
Unlike the sum, computations are needed in addition to concatenation to get this
representation.
The second trick is to speed up these computations by using a Horner-like accumulation
as in \algoName{algo:LowSSSxDense}.
This accumulation will be done on both sides for the computation of all necessary
products \(A_{i,k}B_{k,j}\) where \(A_{i,k}\) is under (resp. over) the diagonal and
\(B_{k,j}\) is over (resp. under) it.

\algoName{algo:SSSxSSS} details these computations, using the \Gi i and \Hi i as accumulators.
It presents an improvement over the algorithm of
\cite[\S 3]{CDGPV03} and \cite[Alg. 7.2]{EG05bsss}:
4 products have been avoided at each step by keeping them in memory in the
\(T_i\) and \(S_i\).
They can also be avoided in the numerical context.

\begin{algorithm}[htb]
\algoCaptionLabel{SSSxSSS}{}
   \begin{algorithmic}[1]
     \algrenewcommand{\algorithmicensure}{\emph{InOut:}} 
     \Require{For both \(M \in \{A, B\}\),
       \(\Tag PM_i,\Tag QM_i,\Tag RM_i,\Tag UM_i,\Tag VM_i,\Tag WM_i\), \(\Tag DM_i\) for appropriate \(i \in \intset{1,N}\)
       an \(s\)-\SSS generator for \(M\)}
     \Ensure{
       A \(2s\)-\SSS generator for \(C = AB\)}
     \Statex \CommentLine{\it All values not given as input are initialised to 0}
     \For {\(i \gets 1\dots N\)}
            \State   \(\Gi{i} \gets \Tag QA_{i - 1} \Tag UB_{i - 1} + T_{i - 1} \Tag WB_{i - 1}\) \label{step:gi}
            \State\( T_i \gets \Tag RA_{i} \Gi{i}\) \label{step:si}
            \State\( S_i \gets \Tag PA_{i} \Gi{i}\) \label{step:ti}
\State \(      Q_i \gets \begin{bmatrix}
\Tag QB_{i}\\
    \Tag QA_{i} \Tag DB_{i} + T_i\Tag VB_{i} 
       \end{bmatrix}\)
\State \( R_i \gets \begin{bmatrix}
    \Tag RB_{i} & 0 \\
       \Tag QA_{i}\Tag PB_i & \Tag RA_i
          \end{bmatrix}\)
\State \(  U_i \gets \begin{bmatrix}
    \Tag UA_{i} & \Tag DA_{i} \Tag UB_{i} + S_i\Tag WB_{i} 
  \end{bmatrix} \)
         \State \( W_i \gets \begin{bmatrix}
\Tag WA_{i} &        \Tag VA_{i}\Tag UB_i \\
0 & \Tag WB_i
          \end{bmatrix}\)
\EndFor
     \For {\(i \gets N \dots 1\)}
     \State \(\Hi{i} \gets \Tag VA_{i + 1} \Tag PB_{i + 1} + T_{i + 1} \Tag RB_{i + 1}\) \label{step:hi}
     \State \(T_i \gets \Tag UA_{i}\Hi{i}\) \label{step:tit}
\State \(          D_i \gets \Tag DA_i \Tag DB_i + S_{i}\Tag VB_{i} + T_i\Tag QB_{i}\) \label{step:xdi}
       \State  \( P^{C}_i \gets \begin{bmatrix}
    \Tag DA_{i} \Tag PB_{i} + T_{i}\Tag RB_{i} & \Tag PA_{i}
  \end{bmatrix} \)
       \State \(T_i \gets \Tag WA_{i}\Hi{i}\)\label{step:tim}
       \State \(      V_i \gets \begin{bmatrix}
    \Tag VA_{i} \Tag DB_{i} + T_i\Tag QB_{i} \\
    \Tag   VB_{i}
      \end{bmatrix}\)
       \EndFor
\State       \Return \(\Call{algo:SssCompression}{ \paren{P_i, Q_i, R_i, U_i, V_i, W_i,  D_i}_{i\in\intset{1,N}}}\)
   \end{algorithmic}
\end{algorithm}

\begin{theorem}
  \label{prop:sxs}
  \algoName{algo:SSSxSSS} computes a \(2s\)-\SSS generator for the product of
  two \(n \times n\) matrices given in \(s\)-\SSS form in
\begin{equation}
  \Time{SSSxSSS}(n,s)=\paren{31 + 2^\omega}C_\omega ns^{\omega-1}.\end{equation}
\end{theorem}
\begin{proof}
  Using Lines \ref{step:gi} and \ref{step:ti} for \(\Gi i\) and Lines \ref{step:hi} and
  \ref{step:tim} for \Hi i, induction on \(i\) shows that
\begin{equation}
  \Gi i =
  \sum_{k = 1}^{i - 1} \Tag RA_{i - 1} \hdots \Tag RA_{k + 1} \Tag QA_k \Tag UB_k \Tag WB_{k + 1} \hdots \Tag WB_{i - 1} 
\end{equation}
\begin{equation}
  \Hi i =
  \sum_{k = i +1}^{N} \Tag WA_{i + 1} \hdots \Tag WA_{k - 1} \Tag VA_k \Tag PB_k \Tag RB_{k - 1} \hdots \Tag RB_{i + 1}
\end{equation}
Combining these results with \cref{step:si} for \(S_i\), \cref{step:tit} for \(T_i\) and
finally \cref{step:xdi}, we get that
\( \Tag DC_i = \sum_{k = 1}^N A_{i,k}B_{k,i} = C_{i,i}\).

When \( i < j\), the products \(A_{i,k}B_{k,j}\) take five shapes: lower block of \(A\)
\(\times\) upper block of \(B\), diagonal block \(\times\) upper block,
upper \(\times\) upper, upper \(\times\) diagonal and upper \(\times\) lower.
The equality
\begin{align}
  \Tag UC_i \Tag WC_{i + 1} \dots \Tag WC_{j - 1} \Tag VC_j = \sum_{k = 1}^N A_{i,k}B_{k,j}
\end{align}
and its
counterpart when \(i > j\) can be checked with tedious but straightforward calculations.

The cost is that of 21 products and 8 sums of \(s \times s\) matrices at each of the
\(N\) steps and on call to \algoName{algo:SssCompression}.
\end{proof}
Again the result of \cref{prop:sxs} is limited to matrices defined on the same grid
and the result always has the same storage size, whatever its quasi-separability order.
This is also true for product with \HSS generators in numerical analysis \cite{SDC07}. 
The \Bruhat format can avoid these issues, but to our knowledge no sub-quadratic algorithm exists
for the product of two \Bruhat generators.
The method used for the sum in \cref{sec:bsum} opens the door towards a linear or quasi-linear product algorithm
using \algoName{algo:LBruhatGen}.

\appendix
\section{Experiments}
\label{app:fig}

We report here on experiments of an implementation of algorithms handling \SSS and \Bruhat
generators over a finite field  in the \texttt{fflas-ffpack}
library~\cite{fflas-ffpack},  at commit
\href{https://github.com/linbox-team/fflas-ffpack/tree/33474b31aaf81487978be06dedbc3a408d4b8bfc}{33474b31aa}.
This library provides efficient dense basic linear algebra routines, such as matrix multiplication,
TRSM and Gaussian elimination revealing the rank profile matrix.
It was compiled with the GNU C++ compiler \texttt{g++} version 9.3.0 and linked with the OpenBLAS
library version 0.3.8\footnote{\url{https://www.openblas.net}}.
The benchmarks are run on a single core of an Intel i5-i7300U@2.6GHz running a Linux Mint-20 system.

For all experiments, the matrices have a fixed dimension \(n = 3000\), over the finite field
\(\Z /131071 \Z\). 
We draw the computation times depending on the quasiseprability orders, on three
type of instances: having a ranks of their upper and lower triangular parts equal to \(1000,
1500\) and \(1750\).

Each point corresponds to the mean of the running times of 50 random instances with same parameters.
\Cref{fig:dtq} compares the running times for the generation from a dense matrix.
\Cref{fig:apq} compares the running times for the product by a random dense \(n\times 500\) block
vector, using the same generators.

\begin{figure}[htpb]
  \centering
  \begin{subfigure}{0.49\textwidth}
    \includegraphics[width = \textwidth]{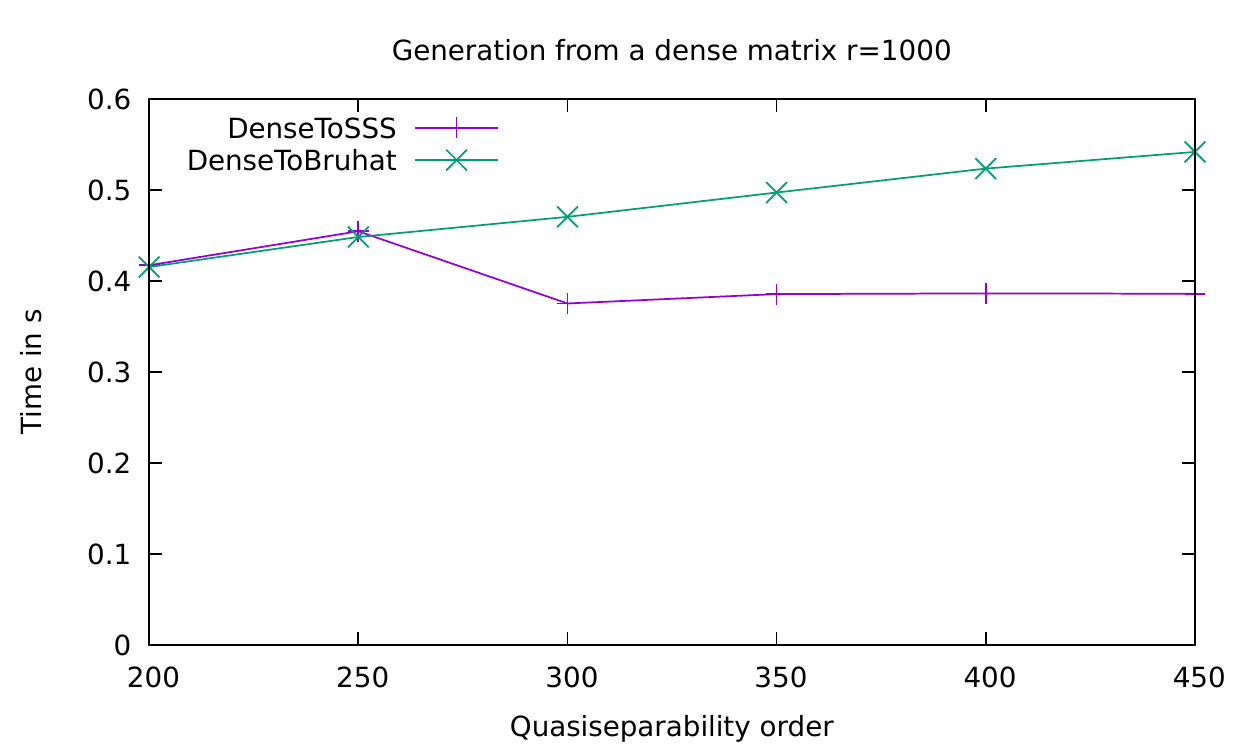}
  \end{subfigure}
  \begin{subfigure}{0.49\textwidth}
    \includegraphics[width = \textwidth]{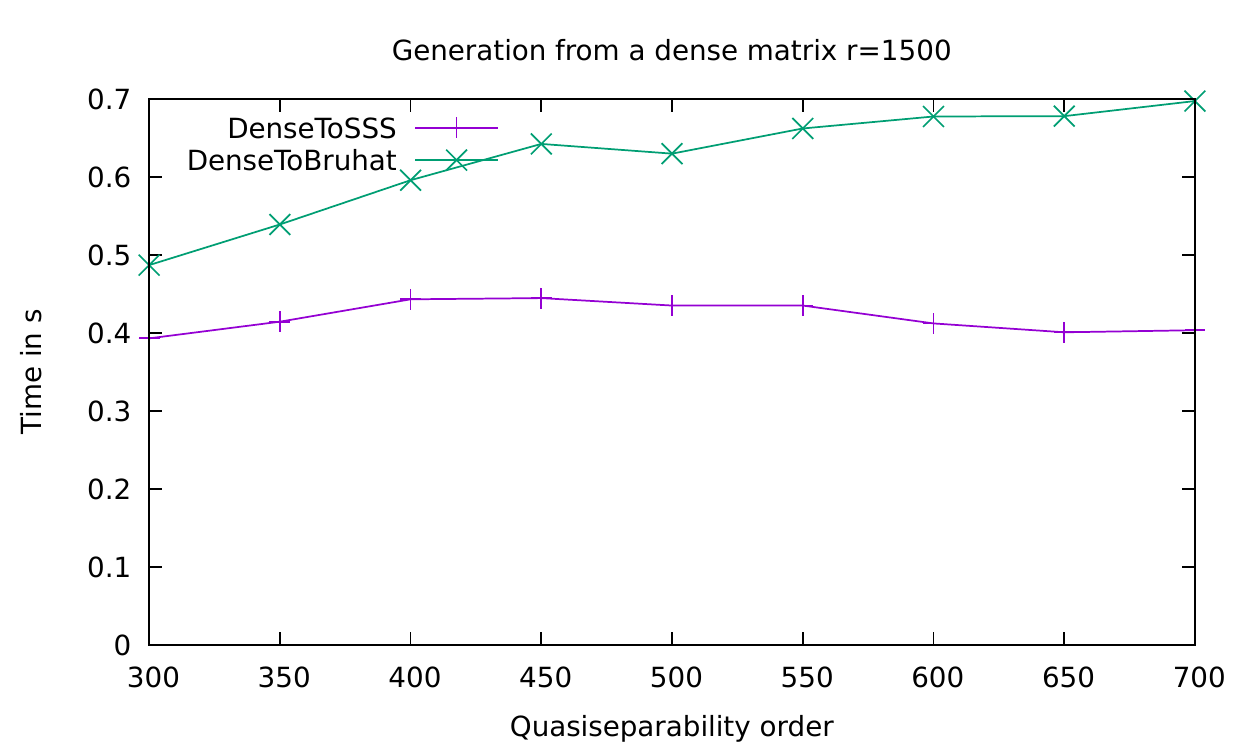}
  \end{subfigure}
  \begin{subfigure}{0.49\textwidth}
    \includegraphics[width = \textwidth]{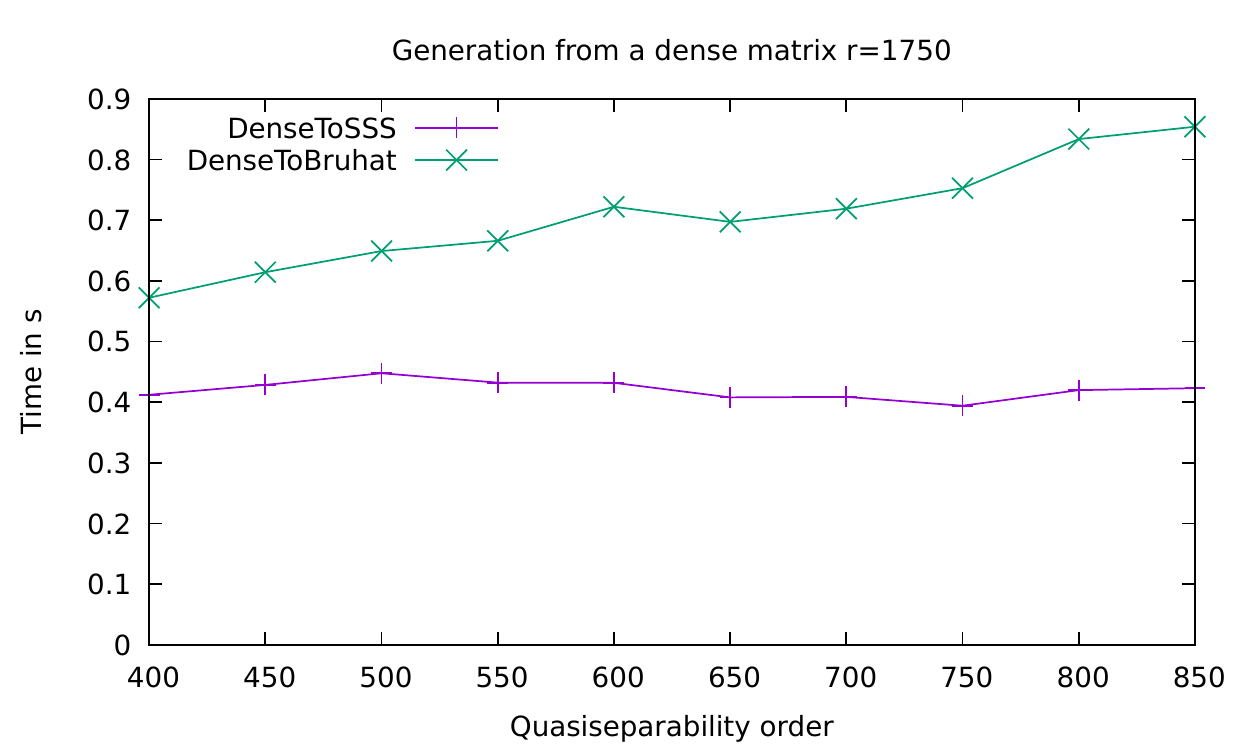}
  \end{subfigure}
  \caption{Experimental timings for the computation of \SSS and \Bruhat generators with \(n = 3000\)
    over \(\Z/131071\Z\)}
  \label{fig:dtq}
\end{figure}
\begin{figure}[htb]
  \centering
  \begin{subfigure}{0.49\textwidth}
    \includegraphics[width = \textwidth]{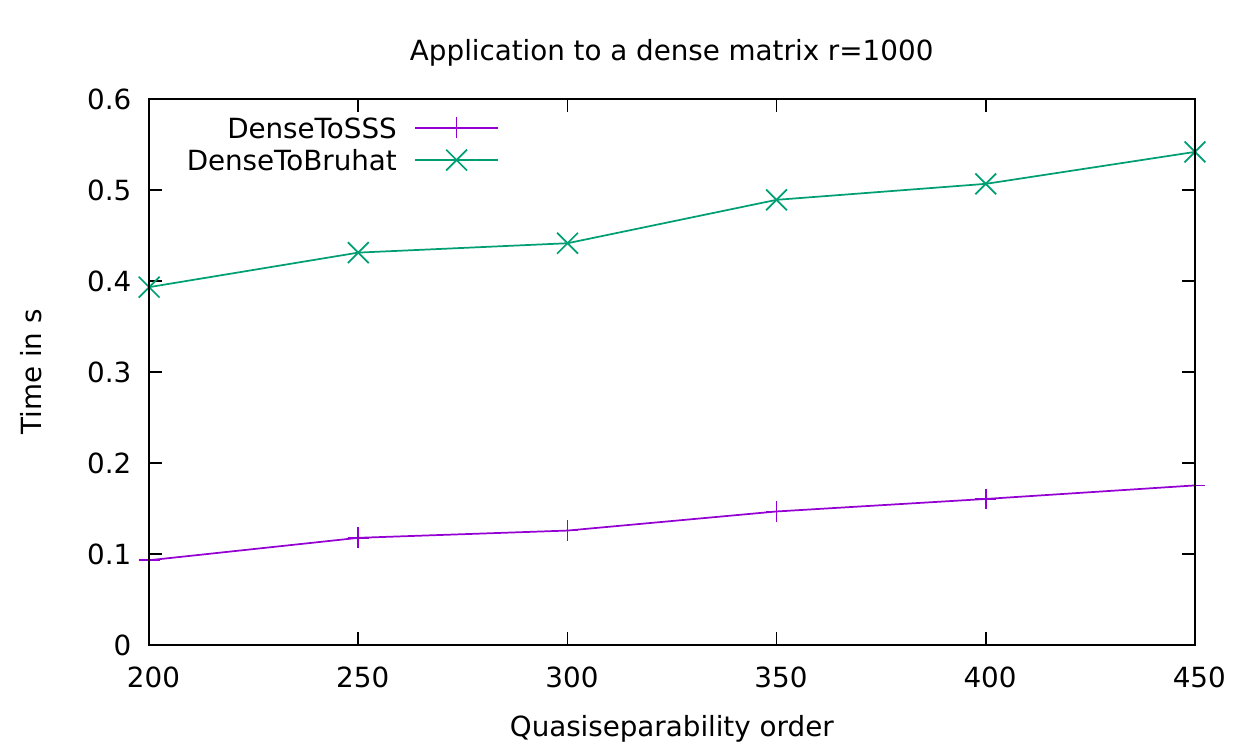}
  \end{subfigure}
  \begin{subfigure}{0.49\textwidth}
    \includegraphics[width = \textwidth]{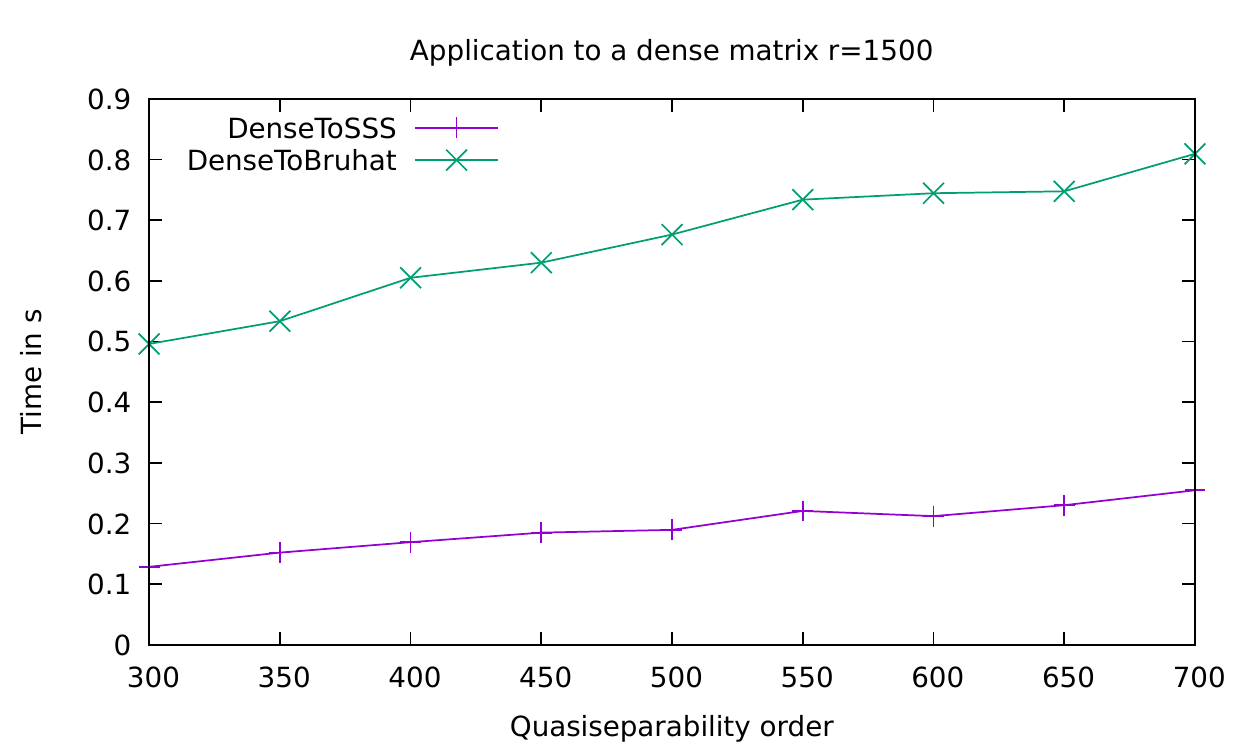}
  \end{subfigure}
  \begin{subfigure}{0.49\textwidth}
    \includegraphics[width = \textwidth]{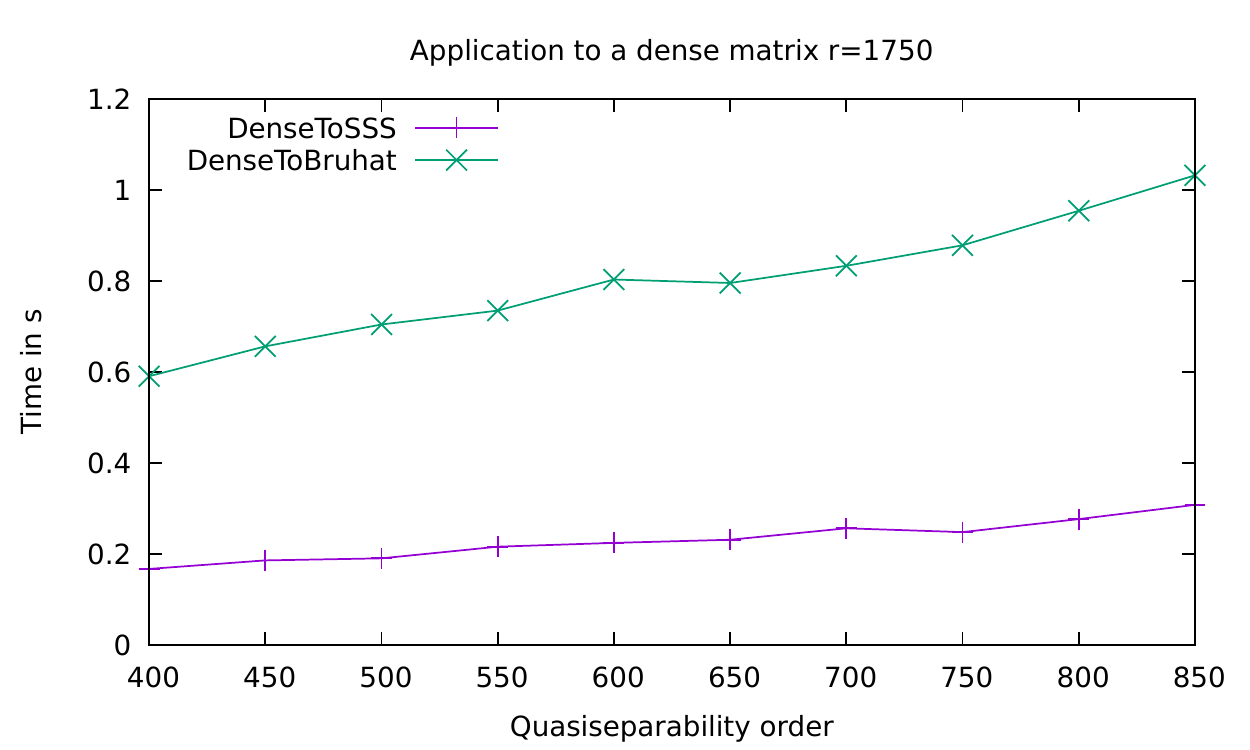}
  \end{subfigure}
  \caption{Experimental timings for the computation of \SSS and \Bruhat times a dense matrix with \(n = 3000\)
  and \(v = 500\) over \(\Z/131071\Z\)}
  \label{fig:apq}
  \end{figure}

\bibliographystyle{ACM-Reference-Format}
\bibliography{qschar}

\end{document}